\documentclass[aps,prl,10pt,twocolumn,superscriptaddress,floatfix]{revtex4-2}
\pdfoutput=1

\usepackage{amsmath}
\usepackage{amssymb}
\usepackage{amsthm}
\usepackage{mathrsfs}
\usepackage{bm, dsfont}
\usepackage[usenames,dvipsnames]{color}
\usepackage{colortbl}
\usepackage[table]{xcolor}
\usepackage{enumitem}
\usepackage{graphicx}
\usepackage{natbib}
\usepackage[colorlinks=true,linkcolor=blue,citecolor=blue,urlcolor=blue]{hyperref}
\usepackage{hypcap}
\usepackage{verbatim}
\usepackage{psfrag}
\usepackage[normalem]{ulem}
\usepackage{physics}
\usepackage[english]{babel}
\usepackage{lineno}
\usepackage{braket}
\usepackage{thmtools, thm-restate}

\declaretheorem{theorem}
\theoremstyle{definition}

\newcommand{\updownarrows}{\uparrow\!\downarrow}
\renewcommand{\paragraph}[1]{\addcontentsline{toc}{section}{#1}\emph{#1.}---}

\begin{document}

\title{Frame dependence of Spekkens’ contextuality for relativistic spin systems}

\author{Ruben Campos Delgado}
\email{ruben.camposdelgado@itp.uni-hannover.de}
\affiliation{Institut für Theoretische Physik, Leibniz Universität Hannover, 30167 Hannover, Germany}

\begin{abstract}
We show that the operational definition of contextuality introduced by Spekkens is, in general, not Lorentz invariant.  Specifically,  we consider an explicit example with particle states consisting of both spin and momentum, we apply a Lorentz transformation to obtain the states in a new inertial frame, and then trace out the momentum degrees of freedom in both frames. We find that, while an observer in the first inertial frame describes a contextual ontological model with respect to spin states and all possible spin measurements, an observer in the boosted frame describes a non-contextual model with respect to the transformed spin states and all transformed spin measurements. Hence, the Spekkens' notion of contextuality, when restricted to spin degress of freedom only, is a frame-dependent concept. We apply our results to predict a novel relativistic effect concerning the task of discriminating between two quantum states. We show that the probability of success for a moving observer exceeds that of an observer at rest. 
\end{abstract}

\maketitle

\paragraph{Introduction}
A fundamental property of quantum mechanics, which distinguishes it from classical physics, is that it is contextual: any assignment of a value to the outcome of the measurement of a physical property must depend on the measurement context, namely on what other properties are simultaneously measured with it. 
In other words, measurements in quantum theory cannot be considered as revealing pre-existing values of a classical hidden variable.  This is the content of the Kochen-Specker theorem \cite{Specker1960, Kochen:1967equ}. The importance of contextuality goes far beyond quantum foundations, as it is also widely considered a resource for quantum computation \cite{Raussendorf:2013, Howard:2014zwm, Bermejo-Vega:2017}, state discrimination \cite{Schmid:2018} and security of quantum key distribution protocols \cite{Horodecki:2010, Singh:2017}.

The original notion of Kochen-Specker contextuality was extended in 2005 by Spekkens \cite{Spekkens:2005} to include unsharp measurements and arbitrary operational theories.  Let $\{P\}$ be the convex set of preparations and $\{M\}$ the convex set of measurements, with outcomes labelled by $k$. In the Spekkens' formulation, the probabilities are reproduced according to the following ontological model:
\begin{equation}\label{eq:basic_def}
    p(k|P,M) =\sum_{\lambda} \mu_P(\lambda)\zeta_{M,k}(\lambda).
\end{equation}
Here, $\mu_P(\lambda)$ are probabilities with $\sum_{\lambda}\mu_P(\lambda)=1$ for all $P$ and $\zeta_{M,k}(\lambda)$ are so-called indicator functions with $\sum_{k}\zeta_{M,k}(\lambda)=1$ for all $M$ and $\lambda$. If both $\mu_P$ and $\zeta_{M,k}$  depend only on the operational equivalence classes, i.e. classes of preparations (respectively measurements) that are not distinguishable by any measurements (respectively states), then the operational theory is said to be non-contextual. In quantum theory, preparation equivalence classes are quantum states $\{\rho_i\}_i$, namely positive semi-definite unit-trace linear operator acting on a finite-dimensional Hilbert space. Equivalence classes of measurements are described by positive operator valued measures (POVM), namely collections of positive semi-definite linear operators $M=\{ M_k\}_k$ which sum to the identity $\sum_k M_k=I$. Furthermore, in quantum theory the probabilities are obtained via the Born rule, $p(k|\rho_i, M)=\Tr(\rho_i M_k)$. Thus, Eq. \eqref{eq:basic_def} translates to \cite{Jokinen:2024qtk}
\begin{equation}\label{eq:Spekkens_def}
   \Tr(\rho_i M_k) = \sum_{\lambda} \Tr(\rho_i G_{\lambda}) \Tr(\sigma_{\lambda} M_k) \hspace{3mm} \forall i, \,\forall M_k, 
\end{equation}
where $\Tr(\rho_i G_{\lambda})\geq 0$, $\Tr(\sigma_{\lambda} M_k)\geq 0$, and $\Tr(\sigma_\lambda) = 1$ for all $\lambda$ and $\sum_{\lambda} \Tr\left(\rho_i G_{\lambda}\right) = 1$ for all $i$. Here, $\{\sigma_{\lambda}\}_\lambda$ are quantum states and $\{G_\lambda\}_{\lambda}$ is a pseudo-POVM, i.e. a POVM that is not necessarily positive semi-definite. 
If a set of states $\{\rho_i\}_i$ satisfies Eq. \eqref{eq:Spekkens_def} for all measurements $M_k$, then we say that the states are non-contextual. Otherwise, they are contextual. 

The first main result of this paper is that the Spekkens' definition is in general not Lorentz invariant when restricted to spin degrees of freedom only. To prove this, we fix an inertial frame and consider an explicit example of particle states with both spin and momentum.  We then apply a Lorentz transformation, in particular a Lorentz boost, and trace out the momentum degrees of freedom in both inertial frames. The ontological model consisting of the chosen spin states and all spin measurements is contextual in the initial frame. However, the model consisting of the transformed spin states and all transformed spin measurements is non-contextual in the boosted frame. Moreover, this is not just a relativistic effect. In fact, the result holds even at low speeds. 
Therefore, the momentum and spin degrees of freedom of massive particles are inherently inseparable, leading to frame-dependent contextuality when analyzed with respect to spin
observables alone. Nevertheless, if one makes the additional assumption that the wavefunctions in momentum space are spherically symmetric, then Lorentz invariance is preserved for one-qubit states. However, this does not hold for states with two or more qubits.

The observed breaking of Lorentz invariance of the Spekkens' definition should not be completely unexpected. In fact, when special relativity is taken into account, it is already known that different observers may disagree on certain physical properties, which are well defined for quantum mechanics alone. For example, thermal radiation that is perfectly black-body in an inertial frame is not thermal if viewed from a moving frame \cite{Peebles:1968zz, Landsberg:1996np}; product states in one frame may appear entangled in another frame \cite{Gingrich:2002ota, Peres:2002ip} with the consequence that the entanglement entropy is not Lorentz invariant \cite{Peres:2002ip}. 

Our result also has consequences for the problem of discriminating between quantum states. First of all, we argue that, in general, contextuality is not necessarily a requirement for the security of quantum key distribution protocols. 
Specifically, we study an idealized version of the BB84 protocol \cite{BennettBrassard2014}. Alice has four pure states involving both spin and momentum. She ignores the momentum degrees of freedom and sends to Bob linearly dependent and contextual spin states. In this first scenario, the protocol is known to be robust against the attacks of an eavesdropper Eve trying to guess the state sent by Alice. In the second scenario, Eve performs a Lorentz boost and observes the transformed states, which are now linearly independent and non-contextual. Eve hopes to use non-contextuality to have more success in guessing the state. However, we explicitly show that her situation now gets worse: the states become so mixed that her probability of success decreases as her speed increases. 
One can then turn to a different task, namely discriminating between only two states, for example two different ensembles of the states sent by Alice.  This time, performing a boost yields an advantage as the probability of success increases. 

The paper is structured as follows:  we first review the proper formalism used to describe the transformation of qubits under a Lorentz transformation. Then we construct an explicit example with one-qubit states where an initial contextual model is transformed to a non-contextual one and prove a theorem regarding the preservation of Lorentz invariance for one-qubit spherically symmetric wavefunctions in momentum space.  In the final part, we use the states constructed in the example to analyze a simple version of the BB84 protocol and the general problem of discriminating between two quantum states. 

\paragraph{Relativistic qubits}
We consider a particle of mass $m$ and four-momentum $p^{\mu}=(E,\textbf{p})$, where $E=\sqrt{m^2+|\textbf{p}|^2}$.  We denote the generic quantum state of the particle as $\ket{ \sigma, p}$, where $\sigma$ is the projection of the spin along the $z$-axis. Under a Lorentz transformation $\Lambda$, the momentum changes in the standard way $p^{\mu}\to\Lambda^{\mu}_{\,\,\nu} p^{\nu}$, whereas the spin undergoes a Wigner rotation \cite{Wigner:1939cj, Weinberg:1995mt}. The transformation of the quantum state, implemented by a unitary operator $U(\Lambda)$, is then
\begin{equation}\label{eq:relativistic_rule}
U(\Lambda)\ket{\sigma,p}=\sum_{\sigma'}D^{(s)}_{\sigma'\sigma}\left[W(\Lambda,p)\right]\ket{\sigma',\Lambda p},
\end{equation}
where $D^{(s)}_{\sigma'\sigma}$ furnishes a representation of the little group for spin $s$ and
\begin{equation}
    W(\Lambda,p):=L^{-1}(\Lambda p) \Lambda L(p),
\end{equation}
known as Wigner rotation, contains a "standard boost" $L(p)$ which maps the four-momentum of a particle at rest $(m,0)$ to a generic $(E,\textbf{p})$. 
From now on, we focus on qubits and set $s=1/2$. The Lorentz transformations we are interested in are boosts and, without loss of generality, we restrict ourselves to a boost with rapidity $\zeta$ in the $z$ direction. After expressing the momentum in spherical coordinates as $p^{\mu} =\left(E,p\sin\theta\cos\varphi, p\sin\theta\sin\varphi,p\cos\theta\right)$,
the matrix form of $D^{(\frac{1}{2})}_{\sigma'\sigma}$ is \cite{Gingrich:2002ota, Halpern1968} 
\begin{equation}\label{eq:wigner_matrix}
D^{(\frac{1}{2})}_{\sigma'\sigma}=\begin{pmatrix} D_{\uparrow\uparrow} && D_{\uparrow\downarrow}\\ D_{\downarrow\uparrow} && D_{\downarrow\downarrow}
\end{pmatrix}=
\begin{pmatrix} \alpha && \beta e^{-i\varphi} \\ -\beta e^{i\varphi} && \alpha\end{pmatrix},
\end{equation}
where
\begin{equation}
    \alpha= \sqrt{\frac{E+m}{E'+m}}\left[\cosh\left(\frac{\zeta}{2}\right)+\frac{p\cos\theta}{E+m}\sinh\left(\frac{\zeta}{2}\right)\right],
\end{equation}
\begin{equation}
    \beta = \frac{p\sin\theta}{\sqrt{(E+m)(E'+m)}}\sinh\left(\frac{\zeta}{2}\right)
\end{equation}
and $E'=E\cosh(\zeta)+p\cos\theta\sinh(\zeta)$. One can verify that $\alpha^2+\beta^2=1$, which ensures that $D^{(\frac{1}{2})}_{\sigma'\sigma}$ is unitary. 

\paragraph{Spekkens contextuality for one-qubit states}
Before looking at a specific example, we state two useful results \cite{Zhang:2025kme, CamposDelgado:2025lyl}. The first connects linear independence of states with non-contextuality.
\begin{theorem} \label{th1}
Let $\mathcal{H}$ be a Hilbert space and let $\mathcal{D}(\mathcal{H})$ be the space of states in $\mathcal{H}$. If the states  $\{\rho_i\}_{i} \subseteq \mathcal{D}(\mathcal{H}) $ are linearly independent, then they are Spekkens non-contextual. 
\end{theorem}
The second result is a necessary and sufficient condition for a set of pure states to be (non-)contextual. 
\begin{theorem} \label{th2}
Let $\mathcal{H}$ be a Hilbert space, let $\mathcal{D}(\mathcal{H})$ be the space of states in $\mathcal{H}$, and let $\{\rho_i\}_{i} \subseteq \mathcal{D}(\mathcal{H})$ be a set of pure states. The set $\{\rho_i\}_{i}$ is Spekkens non-contextual if and only if the density matrices of the states $\rho_i$ are linearly independent. 
\end{theorem}
Notice that the transformation \eqref{eq:relativistic_rule} is linear and invertible. Therefore, it preserves the linear (in)dependence of states.  From the previous theorems, we conclude that Spekkens contextuality is Lorentz invariant when both spin and momentum are taken into account. This is no longer true when one considers spin states only. To illustrate this, we construct an explicit example for one-qubit states. Formally, we start with a total Hilbert space $\mathcal{H}=\mathcal{H}_{\text{spin}}\otimes\mathcal{H}_{\text{momentum}}$. After tracing out the momentum degrees of freedom, we are effectively working only with $\mathcal{H}_{\text{spin}}$. Accordingly, in the definition \eqref{eq:Spekkens_def} as well as in Theorems \ref{th1} and \ref{th2}, we are considering all spin measurements. From now on, whenever we say that a set of states is (non-)contextual, we will always implicitly assume that all spin measurements are considered.  We denote the elements of the computational basis as $\ket{\uparrow}, \ket{\downarrow}$. 
Let us consider in an inertial frame the state
\begin{equation}\label{eq:state1}
\ket{\psi_1}=\int d\mu(p)\, \psi_{\uparrow}(p)\ket{\uparrow,p}
\end{equation}
with the Lorentz invariant measure \cite{Peres:2002wx}
\begin{equation}
    d\mu(p)=\frac{d^3\textbf{p}}{(2\pi)^32E}
\end{equation}
and normalization
\begin{equation}
\langle p|p'\rangle=(2\pi)^3 2E\delta(\textbf{p}-\textbf{p}').
\end{equation}
After tracing out the momentum degrees of freedom, we get the reduced density matrix
\begin{equation}
\rho_1=\Tr_p\ket{\psi_1}\bra{\psi_1}=\int d\mu(p)\, |\psi_{\uparrow}(p)|^2 \ket{\uparrow}\bra{\uparrow}=\ket{\uparrow}\bra{\uparrow}.
\end{equation}
Let us perform a Lorentz boost along the $z$ axis and compute the state in the new inertial frame by applying Eq. \eqref{eq:relativistic_rule}:
\begin{equation}
\begin{gathered}
\ket{\psi'_1}=\int d\mu(p)\,\psi_{\uparrow}(p)U(\Lambda)\ket{\uparrow,p}\\
=\sum_{\sigma'}\int d\mu(p) \,\psi_{\uparrow}(p)D_{\sigma'\uparrow}\left[W(\Lambda,p)\right]\ket{\sigma',\Lambda p}\\
= \sum_{\sigma'} \int d\mu(p)\, \psi_{\uparrow}\left(\Lambda^{-1}p \right)D_{\sigma' \uparrow}\left[W\left(\Lambda,\Lambda^{-1}p\right)\right]\ket{\sigma',p}.
\end{gathered}
\end{equation}
The corresponding reduced density matrix is 
\begin{equation}
\begin{gathered}
    \tau_1=\Tr_p{\ket{\psi'_1}\bra{\psi'_1}}=\int d\mu\, |\psi_{\uparrow}|^2\Big[D_{\uparrow\uparrow}D^*_{\uparrow\uparrow}\ket{\uparrow}\bra{\uparrow}\\+D_{\uparrow\uparrow}D^*_{\downarrow\uparrow}\ket{\uparrow}\bra{\downarrow}
    +D_{\downarrow \uparrow}D^*_{\uparrow\uparrow}\ket{\downarrow}\bra{\uparrow}+D_{\downarrow\uparrow}D^*_{\downarrow\uparrow}\ket{\downarrow}\bra{\downarrow}\Big].
\end{gathered}
\end{equation}
We can replicate the same calculation starting with the states 
\begin{equation}\label{eq:states234}
\begin{gathered}
    \ket{\psi_2}=\int d\mu(p)\, \psi_{\downarrow}(p)\ket{\downarrow,p}, \\
    \ket{\psi_3} = \int d\mu(p)\, \psi_{+}(p)\ket{+,p},\\
    \ket{\psi_4}= \int d\mu(p)\, \psi_{-}(p)\ket{-,p},
\end{gathered}
\end{equation}
where $\ket{\pm,p}=\frac{\ket{\uparrow,p}\pm\ket{\downarrow,p}}{\sqrt{2}}$.
After tracing out the momenta, the reduced density matrices are $\rho_1=\ket{\uparrow}\bra{\uparrow}$, $\rho_2=\ket{\downarrow}\bra{\downarrow},$ $\rho_3=\ket{+}\bra{+}$, $\rho_4=\ket{-}\bra{-}$, which are pure and linearly dependent, hence contextual by Theorem \ref{th2} (with respect to all possible spin measurements). 
It is convenient to define the following integrals:
\begin{equation}\label{eq: integrals_onequbit}
\begin{gathered}
I_{\updownarrows,1}=\int d\mu\,|\psi_{\updownarrows}|^2 D^2_{\uparrow\uparrow}, \hspace{2mm}I_{\updownarrows,2}=\int d\mu\,|\psi_{\updownarrows}|^2 |D_{\downarrow\uparrow}|^2\\
I_{\updownarrows,3}=\int d\mu\,|\psi_{\updownarrows}|^2 D_{\downarrow\uparrow}D_{\uparrow\uparrow},\\
I_{\pm,1}=\int d\mu\, |\psi_{\pm}|^2 D^2_{\uparrow\uparrow}, \hspace{2mm} I_{\pm,2}=\int d\mu\, |\psi_{\pm}|^2 |D_{\downarrow\uparrow}|^2, \\
I_{\pm,3}=\int d\mu\, |\psi_{\pm}|^2 D^2_{\downarrow\uparrow}, \hspace{2mm} I_{\pm,4}=\int d\mu\, |\psi_{\pm}|^2 D_{\downarrow\uparrow}D_{\uparrow\uparrow}. 
\end{gathered}
\end{equation}
The transformed spin states after the Lorentz boost can be simplified by noticing that $D_{\uparrow\uparrow}=D_{\downarrow\downarrow}=D^*_{\uparrow\uparrow}$, $D_{\uparrow\downarrow}=-D^*_{\downarrow \uparrow}$ (see Eq. \eqref{eq:wigner_matrix}), and in matrix form they are given by
\begin{equation}\label{eq:transformed_states}
\begin{gathered}
\tau_1 =
\begin{pmatrix}
I_{\uparrow,1} & I^*_{\uparrow,3} \\
I_{\uparrow,3} & I_{\uparrow,2}
\end{pmatrix}, \quad
\tau_2 =
\begin{pmatrix}
I_{\downarrow,2} & -I^*_{\downarrow,3} \\
- I_{\downarrow,3} & I_{\downarrow,1}
\end{pmatrix}, \vspace{3mm}\\
\tau_3 = \frac{1}{2}
\begin{pmatrix}
- I^*_{+,4}-I_{+,4}+1 & -I^*_{+,3}+I_{+,1} \\
- I_{+,3}+I_{+,1} & I^*_{+,4}+I_{+,4}+1
\end{pmatrix}, \vspace{3mm}\\
\tau_4 = \frac{1}{2}
\begin{pmatrix}
I^*_{-,4}+I_{-,4}+1&  I^*_{-,3}-I_{-,1}\\
 I_{-,3}-I_{-,1}& -I^*_{-,4}-I_{-,4}+1
\end{pmatrix}.
\end{gathered}
\end{equation}
Notice that $\Tr(\tau_i)=1$, as $I_{\uparrow,1}+I_{\uparrow2}=1$, $I_{\downarrow,1}+I_{\downarrow,2}=1$, $I_{\pm,1}+I_{\pm,2}=1$.  In summary, the contextual set of pure,
linearly dependent states $\{\rho_1, \rho_2, \rho_3, \rho_4\}$ is mapped to the set of mixed, linearly independent states $\{\tau_1,\tau_2,\tau_3,\tau_4\}$.
Also note that the initial set has lost information about the form of the original wavefunctions. This information is present in the transformed set. Moreover, we stress the fact that the representation \eqref{eq:wigner_matrix} holds for massive particle states only. For massless particles like photons the Wigner matrix is diagonal \cite{Weinberg:1995mt} with the consequence that some of the integrals in Eq. \eqref{eq: integrals_onequbit} vanish. In other words, Spekkens' contextuality is a Lorentz invariant notion for photonic one-qubit states.
We now show that there exists a pseudo-POVM $F=\{F_j\}_j$ such that $\Tr(\tau_i F_{j})=\delta_{ij}$. 
The general form of a one-qubit linear operator is
\begin{equation}
    F = A\ket{\uparrow}\bra{\uparrow}+B\ket{\uparrow}\bra{\downarrow}+C\ket{\downarrow}\bra{\uparrow}+D\ket{\downarrow}\bra{\downarrow}
\end{equation}
The condition $F=F^{\dagger}$ fixes $A,D\in\mathbb{R}$ and $B=C^*$. Thus, $F(a,b,c,d)=a\ket{\uparrow}\bra{\uparrow}+(b+ic)\ket{\uparrow}\bra{\downarrow}+(b-ic)\ket{\downarrow}\bra{\uparrow}+d\ket{\downarrow}\bra{\downarrow}$. The pseudo-POVM is then composed of $F_1=F(a,b,c,d)$, $F_2=F(a',b',c',d')$, $F_3=F(a'',b'',c'',d'')$ and $F_4=I-F_1-F_2-F_3$, for a total of 12 free real parameters. The number of independent equations is also 12, suggesting that a solution may exist. We verified this by explicitly solving the equations in Mathematica and indeed obtained a solution. The resulting coefficients for the POVM elements are quite lengthy expressions.The corresponding Mathematica notebook is available at the GitHub repository~\cite{RepoKey} or alternatively upon request.

In summary, the model for spin observables in the rest frame is contextual, while the model for spin observables in the boosted frame is not. Note that the two models are different, as they pertain to different sets of observables. 
Physically, the situation can be understood as follows. Alice, in her rest frame, performs arbitrary spin measurements on her spin states and finds that Spekkens’ definition \eqref{eq:Spekkens_def} is not satisfied for any choice of measurement. In other words, her ontological model with spin states $\{\rho_1,\rho_2,\rho_3,\rho_4\}$ and all spin measurements $\{M_{\text{spin}}\}$ is contextual. Bob, on the other hand, applies a Lorentz boost and observes the same physical states, which however appear different in his frame. He finds that, for every possible spin measurement in his frame, the definition is always satisfied. We emphasize that it is not necessary to know the explicit transformation of measurements under a Lorentz boost, since Bob can perform any measurement, including those unrelated to Alice’s.  Indeed, by choosing $G_{\lambda}=F_\lambda$ and $\sigma_\lambda=\tau_\lambda$, $\lambda=1,\cdots,4$, we have
\begin{equation}
\begin{gathered}
    \Tr(\tau_i M_k)=\sum_{\lambda}\Tr(\tau_i G_{\lambda})\Tr(\sigma_{\lambda} M_k)\\=\sum_{\lambda}\Tr(\tau_i F_{\lambda})\Tr(\tau_{\lambda} M_k)=\sum_{\lambda}\delta_{i\lambda}\Tr(\tau_\lambda M_k)\\=\Tr(\tau_i M_k) \hspace{3mm} \forall M_k
\end{gathered}
\end{equation}
In other words, Bob's ontological model with spin states $\{\tau_1,\tau_2,\tau_3,\tau_4\}$ and all transformed spin measurements $\{M'_{\text{spin}}\}$ is non-contextual. 

One can check that, for example, $\lim_{\zeta\to 0}a=\infty$. On the one hand, the divergence of the coefficient recovers the trivial case with no boost, where the states are contextual. On the other hand, it is always possible to have a solution for an arbitrary small $\zeta$, i.e. the final states are non-contextual even in the non-relativistic regime. 

Nevertheless, adding the additional constraint of spherical symmetry of the wavefunctions does make Spekkens' contextuality Lorentz invariant for one-qubit states, as shown by the following theorem. 
\begin{theorem}\label{th3}
If the one-particle wavefunctions in momentum space are spherically symmetric, then Spekkens' contextuality is Lorentz invariant for one-qubit states.
\end{theorem}
\begin{proof}
A generic one-particle mixed state can be written as $\rho=\sum_i p_i \ket{\psi}_i\bra{\psi}_i$, where
\begin{equation}
\ket{\psi}_i = \alpha_i \int d\mu(p)\, \psi_{\uparrow,i}(p)\ket{\uparrow,p}+\beta_i\int d\mu(p)\, \psi_{\downarrow,i}(p)\ket{\downarrow,p} 
\end{equation}
and $\sum_i p_i=1$, $|\alpha_i|^2+|\beta_i|^2=1$ for all $i$ and $\psi_{\uparrow}(p)=\psi_{\uparrow}(|\textbf{p}|)$, $\psi_{\downarrow}(p)=\psi_{\downarrow}(|\textbf{p}|)$.
After tracing out the momentum, the resulting spin state, in matrix form, is
\begin{equation}\label{eq:th_rho}
    \rho = \sum_i p_i\begin{pmatrix}
        |\alpha_i|^2\int d\mu\,|\psi_{\uparrow,i}^2| & \alpha_i \beta^*_i\int d\mu\, \psi_{\uparrow,i}\psi^*_{\downarrow,i}\\
        \alpha^*_i \beta_i\int d\mu\, \psi^*_{\uparrow_i}\psi_{\downarrow_i}&|\beta_i|^2\int d\mu\, |\psi_{\downarrow_i}|^2
    \end{pmatrix}.
\end{equation}
Let us perform a Lorentz boost. Since the wavefunctions are spherically symmetric, they are independent of the angle $\varphi$ so integrals involving $D_{\downarrow\uparrow}$ or $D^2_{\downarrow\uparrow}$ vanish in virtue of $\int_{0}^{2\pi}d\varphi\,e^{i\varphi}=0$. Hence, using again the notation in \eqref{eq: integrals_onequbit}, the resulting final spin state is $\tau=\sum_i p_i \Tilde{\tau}_i$, with
\begin{small}
\begin{equation}\label{eq:th_tau}
    \tilde{\tau}_i = \begin{pmatrix}
    |\alpha_i|^2 {I_{\uparrow,1}}_i+|\beta|^2{I_{\downarrow,2}}_{i}&\alpha_i \beta_i \int d\mu\psi_{\uparrow,i}\psi^*_{\downarrow,i} D^2_{\uparrow\uparrow}\\
    \alpha_i \beta^*_i\int d\mu \psi^*_{\uparrow,i}\psi_{\downarrow,i}D^2_{\uparrow\uparrow} & |\alpha_i|^2{I_{\uparrow,2}}_i+|\beta_i|^2 {I_{\downarrow,1}}_i
    \end{pmatrix}.
\end{equation}
\end{small}
Let us now consider linearly independent states $\{\rho_j\}_j$. Linear independence means that $\sum_j \lambda_j\rho_j=0\to \lambda_j=0 \,\quad\forall j$. In particular, the off-diagonal elements of \eqref{eq:th_rho} give the condition 
\begin{equation}\label{eq:condition}
    \sum_{i}\sum_j{\lambda_{i}}_{j} {\alpha_{i}}_{j} {\beta^*_{i}}_{j} =0 \to {\lambda_i}_j = 0\quad \forall j,
\end{equation}
where ${\lambda_{i}}_{j}=\lambda_j \int d\mu {\psi_{\uparrow,i}}_{j}{\psi^*_{\downarrow,i}}_{j}$. Let us now turn to $\{\tau_j\}_j$. While checking the linear independence or dependence, the resulting equation from the diagonal elements of \eqref{eq:th_tau} is $\sum_i\sum_j\omega_j {\alpha_i}_j{\beta_i}_j\int d\mu {\psi_{\uparrow,i}}_j{\psi^*_{\downarrow,i}}_jD_{\uparrow\uparrow}=0$ or $\sum_{i}\sum_j{\omega_{i}}_{j} {\alpha_{i}}_{j} {\beta^*_{i}}_{j} =0$ where ${\omega_{i}}_{j}=\omega_j \int d\mu {\psi_{\uparrow,i}}_{j}{\psi^*_{\downarrow,i}}_{j}D_{\uparrow\uparrow}$. 
The form of the last condition is the same as the one in Eq. \eqref{eq:condition}. Thus, spin states in one frame are linearly independent if and only if the transformed spin states are linearly independent. By Theorem \ref{th1}, this means that spin states in one frame are non-contextual if and only if they are non-contextual in the other frame. In turn, this implies that spin  states in one frame are contextual if and only if they are contextual in the other frame. 
\end{proof}
The proof of Theorem \ref{th3} suggests that the Lorentz invariance of Spekkens' contextuality for one-qubit states is achieved thanks to the cancellation of unwanted terms. However,  as the dimension of the Hilbert space grows, the matrices representing states with two or more qubits also acquire additional terms and linear independence in one frame does not guarantee linear independence in the other, even for spherically symmetric wavefunctions. 

\paragraph{A relativistic effect in state discrimination}
We consider an idealized version of the BB84 protocol. Alice prepares the four massive states given in Eqs. \eqref{eq:state1}, \eqref{eq:states234} (with non-spherically symmetric wavefunctions), she ignores the momenta degrees of freedom and, with a probability distribution $(p_i)_i$, sends the pure states $\{\rho_1,\rho_2,\rho_3,\rho_4\}=\{\ket{\uparrow}\bra{\uparrow}, \ket{\downarrow}\bra{\downarrow},\ket{+}\bra{+},\ket{-}\bra{-}\}$ to Bob. This situation is robust against the attacks of an eavesdropper Eve. However, Eve knows that a contextual model can be mapped to a non-contextual one and wonders if non-contextuality somehow allows her to break the protocol. She performs a Lorentz boost and in her frame she now observes the mixed states $\{\tau_1,\tau_2,\tau_3,\tau_4\}$ as in Eq. \eqref{eq:transformed_states}. The average probability of success for a given generalised measurement $M=\{M_k\}_k$ to distinguish which state was sent is 
\begin{equation}
    P_{\text{success}}=\sum_i p_i \Tr(\tau_i M_i).
\end{equation}
The optimal measurement is then the solution of the semidefinite program
\begin{equation}
\begin{split}
    &\text{max} \,\,\,\sum_i p_i \Tr(\tau_i M_i)\\
    & \text{s.t.}\hspace{4mm} M_i \succeq 0, \,\sum_i M_i = I.
\end{split}
\end{equation}
The result depends on the explicit form of the wavefunctions $\psi_{\updownarrows}, \psi_{\pm}$. Let us consider the example of a deformed Gaussian:
\begin{equation}
    \psi_{\updownarrows,\pm}=N_{\updownarrows,\pm}\exp\left(-\frac{p^2}{2\sigma_{\updownarrows, \pm}^2}\right)\sqrt{1+\epsilon \cos{\phi}}.
\end{equation}
\begin{figure}
\centering
\includegraphics[width=0.9\linewidth]{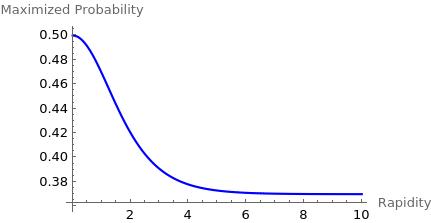}
\caption{Eve's probability of success in discriminating between four states, as a function of the rapidity of the Lorentz boost. Data is shown for $\epsilon=0.1$, $m=1$, $\sigma_{\uparrow}=2$, $\sigma_{\downarrow}=4$, $\sigma_{+}=3$, $\sigma_{-}=6$.  }
\label{fig: fig1}
\end{figure}
The solution of the semidefinite program for the uniform distribution $p_i=1/4$ and the specific choice $\epsilon=0.1$, $m=1$, $\sigma_{\uparrow}=2$, $\sigma_{\downarrow}=4$, $\sigma_{+}=3$, $\sigma_{-}=6$, is shown in Fig. \ref{fig: fig1}. 
We conclude that performing a boost actually worsens the probability of success. The reason behind this can be traced back to the fact that the initial states are pure, while the transformed states are mixed. Thus, the loss of contextuality does not automatically imply an advantage in discriminating between states. However, there is a task in which a moving observer can outperform one at rest.
If we define $\tilde{\rho}_1=\frac{1}{2}(\rho_1+\rho_2)$ and $\tilde{\rho}_2=\frac{1}{2}(\rho_3+\rho_4)$, the probability of successfully discriminating between the two ensembles (i.e. corresponding to the least probability of error) is given by the Helstrom formula \cite{Helstrom1976}:
\begin{equation}
    P_{\text{Helstrom}}=\frac{1}{2}+\frac{1}{4}||\tilde{\rho}_1-\tilde{\rho}_2||_1,
\end{equation}
where $||A||_1=\Tr\left(\sqrt{A^\dagger A}\right)$.
Since $\tilde{\rho}_1=\tilde{\rho}_2$, the probability of success for an eavesdropper at rest is $1/2$. However, the transformed states are different, $\tilde{\tau}_1=\frac{1}{2}(\tau_1+\tau_2) \neq \tilde{\tau}_2=\frac{1}{2}(\tau_3+\tau_4)$ with the consequence that a moving eavesdropper experiences a slightly higher probability, as shown in Fig. \ref{fig: fig2}. 
\begin{figure}
\centering
\includegraphics[width=0.9\linewidth]{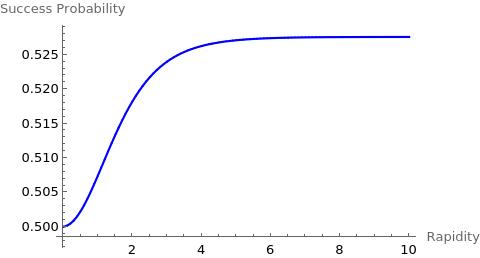}
\caption{Eve's probability of success in discriminating between two ensembles, as a function of the rapidity of the Lorentz boost. Data is shown for $\epsilon=0.1$, $m=1$, $\sigma_{\uparrow}=2$, $\sigma_{\downarrow}=4$, $\sigma_{+}=3$, $\sigma_{-}=6$.  }
\label{fig: fig2}
\end{figure}
In summary, even though performing a Lorentz boost does not guarantee a breaking of quantum key distribution protocols, it does provide an (albeit slight) advantage in the general task of discriminating between two quantum states by means of a Helstrom measurement. 

There is a caveat in our discussion. We have assumed the eavesdropper Eve to have access only to the spin degrees of freedom.  However, in general
cryptographic scenarios one does not restrict the capabilities of eavesdropper. Typically, they are assumed to have access to all
available degrees of freedom, in which case the 
probability of success would not depend on the rapidity. 

\paragraph{Conclusions}
We have shown that, in general, Spekkens’ contextuality restricted to spin states is frame-dependent. We constructed an explicit example for one-qubit states, where we demonstrated that the model for spin observables in the rest frame is contextual, while
the model for spin observables in the boosted frame is not. Importantly, this effect is not limited to relativistic regimes, as the change in contextuality persists even at low speeds. This indicates that, for a truly frame-independent notion of Spekkens’ contextuality, quantum states must include both spin and momentum degrees of freedom.
Nevertheless, experiments restricted to spin degrees of freedom remain meaningful: for instance, observing changes in contextuality can still be done without concern for the inertial frame, as discussed in \cite{CamposDelgado:2025lyl}. Moreover, we showed that imposing spherical symmetry on the momentum-space wavefunction restores Lorentz invariance for one-qubit states, though this does not generalise to multi-qubit systems.
An interesting open question is whether the Spekkens’ framework itself can be generalised to be fully Lorentz invariant. Addressing this question may require techniques from quantum field theory or a careful reconsideration of operational equivalence in relativistic settings.
We also investigated a simple BB84-type scenario in which an eavesdropper performs a Lorentz boost. While boosting generally reduces the success probability in distinguishing four states, we identified a task—discriminating between two ensembles—where a moving observer actually gains an advantage. It would be interesting to identify which quantum information tasks can benefit from relative motion, and to explore how these effects manifest in multi-qubit systems or continuous-variable settings. In conclusion, our work highlights the need to carefully reconsider operational notions, like contextuality, in relativistic settings, potentially opening new avenues for both theory and experiment at the intersection of quantum information and relativity.

\begin{acknowledgments}
\paragraph{Acknowledgments}
The author thanks Martin Plávala, Robert Raussendorf, René Schwonnek and Henrik Wilming for useful discussions. 
\end{acknowledgments}

\bibliography{citations}

\onecolumngrid

\end{document}